\documentclass[12pt,a4paper,english]{amsart}

\usepackage{amsmath,amssymb,amsthm} 

\usepackage{multirow}
\usepackage{tikz}
\usepackage{amsfonts}
\usepackage{caption}
\usepackage{color}

\usepackage{latexsym,mathtools}
\usepackage{epsfig,setspace}
\usepackage{a4}
\usepackage{amssymb}
\usepackage{amsthm}
\usepackage{amsbsy}
\usepackage{upgreek}
\usepackage{relsize}
\usepackage{amsfonts,amssymb,latexsym,amscd,amsmath,euscript,verbatim,calc,enumerate} 
\usepackage[latin1]{inputenc}
\usepackage[cyr]{aeguill}
\usepackage[english]{babel}
\allowdisplaybreaks[1]
\usepackage{url}
\usepackage{hhline}
\usepackage{pifont}
\usepackage[colorlinks=true,linkcolor=blue,citecolor=magenta]{hyperref}
\usepackage{pgf,tikz}
\usepackage{mathrsfs}
\usetikzlibrary{arrows}

\usepackage[style=numeric,firstinits=true,hyperref,useprefix,backend=bibtex]{biblatex}
\bibliography{references_wo_url}
\usepackage{csquotes}

 \theoremstyle{plain}
\newtheorem*{thm}{Theorem}
\newtheorem*{prop}{Proposition}

\newtheorem{theorem}{Theorem}[section]
\newtheorem{lemma}[theorem]{Lemma}
\newtheorem{proposition}[theorem]{Proposition}

\theoremstyle{definition}

\newtheorem{remark}[theorem]{Remark}

\def\FF{{\mathbb F}}

\def\Gal{\mathop{\rm Gal}\nolimits}
\def\Jac{\mathop{\rm Jac}\nolimits}

\def\NS{\mathop{\rm NS}\nolimits}
\def\ev{\mathop{\rm ev}\nolimits}
\def\Num{\mathop{\rm Num}\nolimits}
\def\Pic{\mathop{\rm Pic}\nolimits}

\def\Tr{\mathop{\rm Tr}\nolimits}

\makeatletter
\def\imod#1{\allowbreak\mkern10mu({\operator@font mod}\,\,#1)}
\makeatother

\title{Algebraic geometry codes over abelian surfaces containing no absolutely irreducible curves\\of low genus} 

\author{Yves Aubry, Elena Berardini, Fabien Herbaut and Marc Perret}
\newcommand{\Addresses}{{
  \bigskip
  \footnotesize

  Yves~Aubry, \textsc{Institut de Math\'ematiques de Toulon - IMATH,}\par\nopagebreak
  \textsc{Universit\'e de Toulon and Institut de Math\'ematiques de Marseille - I2M,}\par\nopagebreak
  \textsc{Aix Marseille Universit\'e, CNRS, Centrale Marseille, UMR 7373, France}\par\nopagebreak
  \textit{E-mail address}: \texttt{yves.aubry@univ-tln.fr}

  \medskip

  Elena~Berardini, \textsc{Institut de Math\'ematiques de Marseille - I2M,}\par\nopagebreak
  \textsc{Aix Marseille Universit\'e, CNRS, Centrale Marseille, UMR 7373, France}\par\nopagebreak
  \textit{E-mail address}: \texttt{elena\_berardini@hotmail.it}

  \medskip

  Fabien~Herbaut,  \textsc{INSPE Nice-Toulon,  Universit\'e C\^ote d'Azur,}\par\nopagebreak
  \textsc{Institut de Math\'ematiques de Toulon - IMATH, Universit\'e de Toulon, France}\par\nopagebreak
  \textit{E-mail address}: \texttt{fabien.herbaut@univ-cotedazur.fr}
  
   \medskip
   
  Marc~Perret, \textsc{Institut de Math\'ematiques de Toulouse, UMR 5219,}\par\nopagebreak
  \textsc{Universit\' e de Toulouse, CNRS, UT2J, F-31058 Toulouse, France}\par\nopagebreak
  \textit{E-mail address}: \texttt{perret@univ-tlse2.fr}

}}
\subjclass[2010]{14K05, 14G50, 94B27}
\keywords{Abelian surfaces, AG codes, Weil restrictions, finite fields}
\begin{document}
\maketitle
\begin{abstract}
We provide a theoretical study of 
Algebraic Geometry codes constructed from abelian surfaces 
defined over finite fields.
We give a general bound on their minimum distance and
we investigate how this estimation can be sharpened 
under the assumption that the abelian surface does not contain low genus curves. 
This approach naturally leads us to consider Weil restrictions of elliptic curves and 
abelian surfaces which do not admit a principal polarization.
\end{abstract}
\section{Introduction}\label{Intro}

The success 
of Goppa construction (\cite{Goppa}) of codes over algebraic curves 
in breaking the Gilbert-Varshamov bound (see Tsfasman-Vl\u{a}du\c{t}-Zink bound in \cite{Tsfasman_Vladut_Zink})
has been generating much interest over the last forty years.
This gave birth to the field of Algebraic Geometry codes.
It results a situation with 
a rich background and many examples of 
evaluation codes derived from algebraic curves  (see for instance \cite{Tsfasman_Vladut_Nogin}).
The study of Goppa construction
 from higher dimensional varieties
 has begun with few exceptions
in the first decade of the twenty-first century.
Although the construction holds in any dimension, the main focus has been put on algebraic surfaces.

The case of ruled surfaces is considered by Aubry in \cite{Aubry}.
The case of toric surfaces is addressed among others by Little and Schenck in \cite{Little_Schenck_Toric} and by Nardi in \cite{Nardi}. 
Voloch and  Zarzar 
introduce the strategy of looking 
for surfaces with small Picard number (\cite{Voloch_Zarzar} and \cite{Zarzar}).
This approach is discussed in \cite{Little_Schenck} and used by Couvreur in
\cite{Couvreur} to obtain very good codes over rational surfaces.
In a parallel direction
Little and Schenck (\cite{Little_Schenck}) stress 
the influence of the sectional genus of the surface, 
that is the genus of a generic section.
Finally, Haloui investigates
the case of simple Jacobians of curves of genus $2$ 
in \cite{Haloui}.

\medskip

The aim of this article is to study 
codes constructed from general abelian surfaces. 
While from the {\sl geometric} point of view
(i.e.\ over an algebraically closed field)
a principally polarized abelian surface 
is isomorphic either to the Jacobian of a curve of genus $2$ or to the product of two elliptic curves, 
the landscape turns to be richer from the {\sl arithmetic} point of view. 
Weil proved that 
over a finite field $k$
there is exactly one more possibility, that is
the case of the Weil restriction 
of an elliptic curve defined over a quadratic extension of $k$ (see for instance \cite[Th.1.3]{Howe_Nart_Ritzenthaler}).
Moreover, one can also consider 
abelian surfaces which do not admit a principal polarization.

The main contribution of this paper is twofold. 
First, we give a lower bound on the minimum distance of codes constructed over general abelian surfaces.
Secondly, we sharpen this lower bound
for abelian surfaces which do not contain absolutely irreducible curves defined over $\mathbb{F}_{q}$ of arithmetic genus less or equal than a fixed integer $\ell$.
In order to summarise our results in the following theorem,
let us consider an ample divisor $H$
on an abelian surface $A$ 
and let us denote by
$\mathcal{C}(A,rH)$  
the generalised evaluation code whose
construction is recalled in Section \ref{Codes_from_Abelian_Surfaces}.

\begin{thm}(Theorem \ref{ourbound} and Theorem \ref{minimum_distance_simple_surfaces}) \label{Thm_Intro}
Let $A$ be an abelian surface defined over $\mathbb{F}_q$ of trace $\Tr(A)$. 
Let $m=\lfloor 2\sqrt{q}\rfloor$, $H$ be an ample divisor on $A$
rational over $\mathbb{F}_q$ and $r$ be a positive integer large enough  so that $rH$ is very ample.

Then,
the minimum distance $d(A,rH)$ of the code $\mathcal{C}(A,rH)$ satisfies
\begin{equation} \label{bound_1}
d(A,rH)\geq  \#A(\mathbb{F}_q)-rH^2(q+1-\Tr(A)+m)-r^2m\frac{H^2}{2}. 
\end{equation}
Moreover, if $A$ is simple and contains no absolutely irreducible curves of arithmetic genus $\ell$ or less for some positive integer $\ell$, then
\begin{equation}\label{mainbound}
d(A,rH) \geq  \#A(\mathbb{F}_q)-\max\left( \left\lfloor r\sqrt{\frac{H^2}{2}}\right\rfloor(\ell-1), \varphi(1), \varphi\left(\left\lfloor r\sqrt{\frac{H^2}{2\ell}}\right\rfloor \right)\right),
\end{equation}
where
\footnotesize{
$$\varphi(x):= m\left(r \sqrt{\frac{H^2}{2}} - x \sqrt{\ell}\right)^2+2m\sqrt{\ell}\left(r \sqrt{\frac{H^2}{2}} - x \sqrt{\ell}\right)  +x \Big{(} q+1-\Tr(A)+(\ell-1)(m-\sqrt{\ell})\Big{)} + r \sqrt{\frac{H^2}{2}}(\ell-1).$$}
\end{thm}
If $A$ is simple, then we can take $\ell = 1$ and the lower bound~(\ref{mainbound}) is nothing but Haloui's one~\cite{Haloui} proved in the case of simple Jacobian surfaces $\Jac(C)$ with the choice $H=C$ (see Remark \ref{comparing_with_Haloui}).

\medskip
 It is worth to notice that the second bound is better for larger $\ell$ (at least for $q$ sufficiently large and $1<r<\sqrt{q}$, see Remark~\ref{mieux_ell}). In particular, the bound obtained for $\ell=2$ improves the one obtained for $\ell=1$.
This leads us to 
investigate the case of abelian surfaces 
with no absolutely irreducible curves of genus 1 nor 2, which are necessarily either Weil restrictions of elliptic curves on a quadratic extension, either not principally polarizable abelian surfaces, from the classification given above. 
The following proposition
lists all situations for which we can apply
bound~(\ref{mainbound}) with $\ell=2$.
The key point of the proof 
is a characterisation of isogeny classes of abelian surfaces
containing Jacobians of curves of genus $2$ 
obtained by Howe, Nart and Ritzenthaler (\cite{Howe_Nart_Ritzenthaler}).

\begin{prop}(Proposition \ref{nppas} and Proposition \ref{without_curves}) 
The bound on the minimum distance (\ref{mainbound}) of the previous theorem holds when taking  $\ell=2$ 
in the two following cases:
\begin{enumerate}[(i)]
\item Let $A$ be an abelian surface defined over $\mathbb{F}_q$ which does not admit a principal polarization. Then $A$ does not contain absolutely irreducible curves of arithmetic genus $0$, $1$ nor $2$.

\item Let $q$ be a power of a prime $p$. Let $E$ be an elliptic curve defined over $\mathbb{F}_{q^2}$ of Weil polynomial $f_{E/\mathbb{F}_{q^2}}(t)=t^2-\Tr(E/\mathbb{F}_{q^2})t+q^2$. Let $A$ be the $\mathbb{F}_{q^2}/\mathbb{F}_{q}$-Weil restriction of the elliptic curve $E$. Then $A$ does not contain absolutely irreducible curves defined over $\mathbb{F}_{q}$  of arithmetic genus $0$, $1$ nor $2$ if and only if one of the following cases holds:
\begin{enumerate}[(1)]
\item $\Tr(E/\mathbb{F}_{q^2})=2q-1$;
\item $p>2$ and $\Tr(E/\mathbb{F}_{q^2})=2q-2$;
\item $p \equiv 11 \mod 12$ or $p=3$, $q$ is a square and $\Tr(E/\mathbb{F}_{q^2})=q$;
\item $p=2$, $q$ is nonsquare and $\Tr(E/\mathbb{F}_{q^2})=q$;
\item\label{five} $q=2$ or $q=3$ and $\Tr(E/\mathbb{F}_{q^2})=2q$.
\end{enumerate}
\end{enumerate}
\end{prop}

The paper is structured as follows.
In Section \ref{Codes_from_Abelian_Surfaces} 
we consider evaluation codes 
on general abelian surfaces. 
We compute their dimension and prove the lower bound~(\ref{bound_1}) on their minimum distance.
Section \ref{Codes_from_Simple_Abelian_Surfaces} is devoted to 
the case of simple abelian surfaces,
 that is those for which we can choose some $\ell \geq 1$. 
We derive the lower bound~(\ref{mainbound}) 
depending on the minimum arithmetic genus of absolutely irreducible curves lying on the surface.
In Section \ref{Weil_restriction_of_elliptic_curves} we consider
abelian surfaces which do not admit a principal polarization and
Weil restrictions of elliptic curves to find all
abelian surfaces defined over a finite field containing no absolutely irreducible curves of arithmetic genus $0$, $1$ and $2$. 
Finally, in Section \ref{explicit}, we make explicit the lower bounds for the minimum distance.

\section{Codes from abelian surfaces}\label{Codes_from_Abelian_Surfaces}

\subsection{Some facts on intersection theory}
One of the ingredients of the proofs of
Theorems \ref{ourbound} and \ref{minimum_distance_simple_surfaces}
is the classical inequality induced by the Hodge index theorem (\ref{Hodge})
in the context of intersection theory on surfaces. 
In this subsection, we briefly recall this context and the main properties
we need.
We refer the reader to \cite[\S V]{Hartshorne} for further details.

Let $X$ be a nonsingular, projective, absolutely irreducible 
algebraic surface defined over ${\mathbb F}_q$. A divisor on $X$ is an element
of the free abelian group generated by the  irreducible curves on $X$. Divisors associated to  rational functions on $X$ are called principal. Two divisors on $X$ are said to be {\sl linearly equivalent} if their difference is a principal divisor. We write $\Pic(X)$ for the group of divisors of $X$ modulo linear equivalence. The N\'eron-Severi group of $X$, denoted by $\NS(X)$, is obtained by considering the coarser {\sl algebraic equivalence} we do not define here since it coincides for abelian varieties (see \cite[\S IV]{Lang}) with the following {\sl numerical equivalence}.
A divisor $D$ on $X$ is said to be numerically equivalent to zero, which we denote by $D\equiv 0$, if the intersection product $C.D$ is zero for all curves $C$ on $X$. This gives the coarsest equivalence relation on divisors on $X$ and we denote the group of divisors modulo numerical equivalence by $\Num(X)$.
We have thus $\Num(X)=\NS(X)$, so we will refer to these two equivalence relations with no distinction. We write simply $D$ for the class of a divisor $D$ in $\NS(X)$.

We recall the Nakai-Moishezon criterion in the context of surfaces: a divisor $H$ is ample if and only if $H^2>0$ and  $H.C>0$ for all irreducible curves $C$ on $X$ (\cite[\S V, Th.1.10]{Hartshorne}).
The Hodge index theorem states
 that the intersection pairing is negative definite on the orthogonal complement
 of the line generated by an ample divisor. From this, it easily follows that 
\begin{equation}\label{Hodge} 
H^2D^2 \leq (H.D)^2
\end{equation}
for any pair of divisors $D$, $H$ with $H$ ample, and that equality holds if and only if $D$ and $H$ are numerically proportional.
\subsection{Evaluation codes}\label{section-evcode}
This subsection begins 
by a reminder about
definitions of the evaluation code 
we study.
To this end
we consider again
$X$ a nonsingular, projective, absolutely irreducible 
algebraic surface defined over ${\mathbb F}_q$
and $G$ a divisor on $X$.
The Riemann-Roch space $L(G)$ is defined by
$$L(G)=\{f\in\mathbb{F}_q(X)\setminus \{0\} \ \vert \ (f)+G\geq 0 \}\cup \{0\}.$$
The Algebraic Geometry
code $\mathcal{C}(X,G)$ is sometimes presented 
from a functional point of view
as the image of the following linear evaluation map $\ev$
\begin{center}
\begin{align*}
\ev : \quad L(&G) \longrightarrow \mathbb{F}^n_q&\\
&f\;\; \longmapsto (f(P_1),\dots,f(P_n))&
\end{align*}
\end{center}
which is clearly well defined when considering
 $\{P_1,\dots,P_n\} \subset X(\mathbb{F}_q)$ a subset of rational points 
which are on $X$ but not in the support of $G$.
In fact, this construction 
naturally extends to the case where
$\{P_1,\dots,P_n\} =X(\mathbb{F}_q)$ is an enumeration 
of the whole set of the rational points on $X$,
 as noticed by Manin and Vl\u{a}du\c{t} 
 in \cite[\S 3.1]{Manin}.
Indeed, one can rather consider the image of the following map,
where we denote by $ \mathcal{L}$ the line bundle associated to $L(G)$,
by $\mathcal{L}_{P_i}$ the stalks at the $P_i$'s, and by $s_{P_i}$ 
the images of a global section $s \in H^0\left( X,  \mathcal{L} \right)$ in the stalks
\begin{center}
\begin{align*}
\ev : H^0\left( X,  \mathcal{L} \right)  \longrightarrow \bigoplus_{i=1}^n \mathcal{L}_{P_i} = \mathbb{F}^n_q& \\
s\;\; \longmapsto (s_{P_1},\dots,s_{P_n}) .&
\end{align*}
\end{center}
Different choices of isomorphisms between 
the fibres $\mathcal{L}_{P_i}$ and $\mathbb{F}_q$ give rise to different maps 
but lead to equivalent codes.
See also \cite{Lachaud} or \cite{Aubry}
for another constructive point of view.

Throughout the whole paper we  associate
to a nonzero
function $f\in L(G)$ an effective rational divisor
\begin{equation}\label{decomposition}
D:=G+(f)=\sum_{i=1}^k n_i D_i,
\end{equation}
where $n_i>0$ and  where
each $D_i$ is an ${\mathbb F}_q$-irreducible curve whose arithmetic genus is denoted by $\pi_i$. 
The evaluation map $\ev$ is injective if and only if
  the number $N(f)$ of zero coordinates of the codeword $\ev(f)$ satisfies
\begin{equation}\label{injective}
N(f)<\# X({\mathbb F}_q)
\end{equation}
for any $f \in L(G) \setminus \{ 0 \} $.
In this case the minimum distance $d(X,G)$ of the code $\mathcal{C}(X,G)$ satisfies
\begin{equation}\label{minimumdistance} 
d(X,G) = \#X(\mathbb{F}_q)-\max_{f\in L(G)\setminus \{0\}} N(f).
\end{equation}
Let us remark now that by (\ref{decomposition}) we have 
\begin{equation}\label{Nf_sum} 
N(f)\leq \sum_{i=1}^k \# D_i(\mathbb{F}_q)
\end{equation}
for any $f\in L(G) \setminus \{ 0 \}$.
Therefore, to get a lower bound on the minimum distance of the code $\mathcal{C}(X,G)$ 
it suffices to get two upper bounds:
\begin{itemize}
\item  an upper bound
on the number $k$ of ${\mathbb F}_q$-irreducible components 
of an effective divisor linearly equivalent to $G$
$$D=\sum_{i=1}^k n_i D_i \sim G;$$
\item an upper bound on 
the number of rational points on each ${\mathbb F}_q$-irreducible curves $D_i$ in the support of $D$.
\end{itemize}

\subsection{The parameters of codes over abelian surfaces}
In this subsection we begin 
the estimation of the parameters
of the code in the context of our work.

Let $A$ be
an abelian surface defined over ${\mathbb F}_q$. We recall that the Weil polynomial of an abelian variety is the characteristic polynomial of the Frobenius endomorphism acting on its Tate module. Since $A$ is here 
two-dimensional, 
it has by Weil theorem the shape
\begin{equation}\label{weil_polynomial}
f_A(t)=t^4-\Tr(A)t^3+a_2t^2-q\Tr(A)t+q^2.
\end{equation}
By the Riemann Hypothesis $f_A(t)=(t-\omega_1)(t-\overline{\omega}_1)(t-\omega_2)(t-\overline{\omega}_2)$ where $\omega_i$ are complex numbers of modulus $\sqrt{q}$.
The number $\Tr(A)=\omega_1+\overline{\omega}_1+\omega_2+\overline{\omega}_2$ is called the {\sl trace of $A$}.

Let  $H$ be
 an ample divisor on $A$ rational over $\mathbb{F}_q$ and $r$ large enough  so that $rH$ is very ample
($r\geq 3$ is sufficient by \cite[III, \S 17]{Mumford}). 
Our goal is
 to derive from (\ref{minimumdistance})
 a lower bound on the minimum distance of the code $\mathcal{C}(A,rH)$. 

If the evaluation map $\ev$ is injective, then the dimension of $\mathcal{C}(A,rH)$ is equal to the dimension $\ell(rH)$ of the Riemann-Roch space $L(rH)$ 
which can be computed using 
the Riemann-Roch theorem for surfaces. 
In the general setting of a divisor $D$ on a surface $X$
it states that (see \cite[V, \S 1]{Hartshorne})
$$\ell(D)-s(D)+\ell (K_X-D)=\frac{1}{2}D.(D-K_X)+1+p_a(X)$$
where $K_X$ is the canonical divisor of $X$ and $p_a(X)$ is the arithmetic genus of $X$, and where $s(D)=\dim_{\mathbb{F}_q} H^1(X,\mathcal{L}(D))$ is the so-called {\sl superabundance} of $D$ in $X$.

Since $A$ is an abelian surface we have
(\cite[III, \S 16]{Mumford}) $K_A=0$ and $p_a(A)=-1$. Moreover, if $rH$ is very ample, then we can deduce from \cite[V, Lemma 1.7]{Hartshorne} that $\ell (K-rH)=\ell (-rH)=0$ and that $s(rH)=0$ (\cite[III, \S 16]{Mumford}). 
So finally if the evaluation map $\ev$ is injective, i.e.\ if inequality (\ref{injective}) holds, we get the dimension of the code $\mathcal{C}(A,rH)$:
$$\dim_{\mathbb{F}_q} L(rH)=r^2\frac{H^2}{2}.$$

We are now going to give a lower bound on the minimum distance of $\mathcal{C}(A,rH)$ using (\ref{minimumdistance}) and (\ref{Nf_sum}).
Theorem 4 of \cite{Haloui} states that the number of rational points on a projective ${\mathbb F}_q$-irreducible curve $D$ defined over $\mathbb{F}_q$ of arithmetic genus $\pi$ lying on an abelian surface $A$ of trace $\Tr(A)$ is bounded by
$$\#D(\mathbb{F}_q)\leq q+1-\Tr(A)+|\pi-2| \lfloor 2\sqrt{q}\rfloor.$$
Hence, if we set $m:=\lfloor 2\sqrt{q}\rfloor$, where $\lfloor x\rfloor$ denotes the integer part of the real number $x$, from inequality (\ref{Nf_sum}) we get 
\begin{equation}\label{Nf}
N(f)\leq k(q+1-\Tr(A))+m \sum_{i=1}^k |\pi_i-2|.
\end{equation}
With no hypotheses on the abelian surface nor on the arithmetic genera $\pi_i$, we can only say that $\pi_i = 0$ cannot occur and since $\pi_i\geq |\pi_i-2|$ for $\pi_i\geq1$, we have
\begin{equation}\label{Nf_toutes}
N(f)\leq k(q+1-\Tr(A))+m \sum_{i=1}^k \pi_i.
\end{equation}
In order to use (\ref{Nf_toutes}) to bound the minimum distance of the code $\mathcal{C}(A,rH)$, we need Lemma \ref{HodgeTGV} below, giving upper bounds on the number $k$ of irreducible components of the effective divisor $D$ linearly equivalent to $rH$ and on the sum of the arithmetic genera of its components $D_i$.
We recall for this purpose  a generalisation of the adjunction formula which states that for a curve $D$ of arithmetic genus $\pi$ on a surface $X$ we have $D.(D+K_X)=2\pi-2$ (\cite[\S V, Exercise 1.3]{Hartshorne}). In the case of an abelian surface $A$ for which $K_A=0$, this says that for any curve $D$ of arithmetic genus $\pi$ lying on $A$ we have $D^2=2\pi-2$. 

\begin{lemma}\label{HodgeTGV}

Let $D$ be an effective divisor linearly equivalent to $rH$, let $D=\sum^k_{i=1} n_iD_i$
be its decomposition as a sum of ${\mathbb F}_q$-irreducible curves and let $\pi_i$ be the arithmetic genus of $D_i$ for $i=1,\dots,k$. Then we have

$$\sum_{i=1}^k \pi_i\leq r^2\frac{H^2}{2}+k \quad\text{ and  }\quad k\leq rH^2.$$

\end{lemma}

\begin{proof}
Applying Formula (\ref{Hodge}) to $H$ and $D_i$ for every $i$, we get
$D_i^2H^2\leq(D_i.H)^2 .$
By the adjunction formula we have

\begin{equation}\label{sum_pi}
\pi_i-1\leq (D_i.H)^2/(2H^2).
\end{equation}
Indeed $H^2> 0$ by the Nakai-Moishezon criterion since $H$ is ample.

Summing from $i=1$ to $k$, we obtain
\begin{equation}\label{cor_11}
\sum_{i=1}^k \pi_i-k\leq \frac{1}{2H^2}\sum_{i=1}^k(D_i.H)^2.
\end{equation}
We have also
\begin{equation}\label{TGV}
\begin{split}
\sum_{i=1}^k(D_i.H)^2&=\left(\sum_{i=1}^kD_i.H \right)^2-\sum^k_{i\neq j}(D_i.H)(D_j.H)\\
&\leq \left(\sum_{i=1}^kn_iD_i.H \right)^2-\sum_{i\neq j}^k(D_i.H)(D_j.H)\\
&\leq r^2(H^2)^2,
\end{split}
\end{equation}
where we used the facts that $n_i>0$, that $D=\sum_{i=1}^k n_i D_i$ is linearly (and hence numerically) equivalent to $rH$ and that $D_i.H>0$ for every $i=1,\dots,k$, thanks to Nakai-Moishezon criterion since $H$ is ample.
Now applying inequality (\ref{TGV}) to inequality (\ref{cor_11}), we get
$$  \sum_{i=1}^k \pi_i \leq \frac{r^2(H^2)^2}{2H^2}+k=\frac{r^2H^2}{2}+k$$
which completes the proof of the first statement.
Using that $k\leq \sum_{i=1}^k n_i D_i.H=rH^2$ we get the second one.

\end{proof}

As a consequence of Lemma \ref{HodgeTGV} we can state the following theorem.

\begin{theorem}\label{ourbound}
Let $A$ be an abelian surface defined over $\mathbb{F}_q$ of trace $\Tr(A)$. 
Let $m=\lfloor 2\sqrt{q}\rfloor$, $H$ be an ample divisor on $A$
rational over $\mathbb{F}_q$ and $r$ be a positive integer large enough  so that $rH$ is very ample.

Then the minimum distance $d(A,rH)$ of the code $\mathcal{C}(A,rH)$ satisfies
$$
d(A,rH)\geq  \#A(\mathbb{F}_q)-rH^2(q+1-\Tr(A)+m)-r^2m\frac{H^2}{2}. 
$$
\end{theorem}
\begin{proof}
Using Lemma~\ref{HodgeTGV} together with (\ref{Nf_toutes}) we get $N(f)\leq \phi(k)$
with $\phi(k):=k(q+1-\Tr(A)+m)+mr^2H^2/2 \text{ and } k\in[1,rH^2].$
This means that $N(f)\leq \max_{k\in[1,rH^2]}\{ \phi(k)\}$.
Now remark that 
$ \phi$ is an increasing linear function
since $\vert\Tr(A)\vert \leq 4\sqrt{q}$, and hence gets its maximum when $k=rH^2$. Therefore we have $N(f)\leq \phi\left(rH^2\right)$, which implies $d(A,rH)= \#A(\mathbb{F}_q)-\max \{N(f), f\in L(rH)\setminus \{0\}\} \geq \#A(\mathbb{F}_q)-\phi\left(rH^2\right)$.
The theorem is proved since $ \phi\left(rH^2\right)=rH^2(q+1-\Tr(A)+m)+mr^2H^2/2$.
\end{proof}

\begin{remark}
Let $H$ be an ample divisor. Suppose that $H$ is irreducible over $\mathbb{F}_q$, but reducible on a Galois extension of prime degree $e$. Then $H$ is a sum of $e$ conjugate irreducible components such that the intersection points are also conjugates under the Galois group. Then, by Lemma 2.3 of \cite{Voloch_Zarzar}, we have
$$k\leq r\frac{H^2}{e}.$$
Hence under this hypothesis we get a sharper bound on the number of irreducible components of a divisor linearly equivalent to $rH$, thus a sharper bound for Theorem \ref{ourbound}.
\end{remark}

\section{Codes from abelian surfaces with no small genus curves}\label{Codes_from_Simple_Abelian_Surfaces}

We consider now evaluation codes $\mathcal{C}(A,rH)$ on abelian surfaces which contain no absolutely irreducible curves defined over $\mathbb{F}_q$ of arithmetic genus smaller than or equal to an integer $\ell$.

Throughout this section $A$ denotes a {\sl simple} abelian surface defined over $\mathbb{F}_q$.
Let us remark that by Proposition 5 of \cite{Haloui} a simple abelian surface contains no irreducible curves of arithmetic genus $0$ nor  $1$  defined over $\mathbb{F}_q$.
In particular, every absolutely irreducible curve on $A$ has arithmetic genus greater than or equal to $2$ and thus it is relevant to take $\ell\geq 1$.

\begin{lemma}\label{Borne_N(f)}
Let $A$ be a simple abelian surface defined over $\mathbb{F}_q$ of trace $\Tr (A)$. Let $\ell$ be a positive integer such that for every absolutely irreducible curves of arithmetic genus $\pi$ lying on $A$ we have $\pi>\ell$. Let $f$ be a nonzero function in $L(rH)$ with associated effective rational divisor $D=\sum_{i=1}^k n_iD_i$ as given in equation (\ref{decomposition}). Write $k=k_1+k_2$ where $k_1$ is the number of $D_i$ which have arithmetic genus $\pi_i > \ell$ and $k_2$ is the number of $D_i$ which have arithmetic genus $\pi_i \leq \ell$. Then
\begin{equation}\label{Nf_simple}
N(f)\leq k_1(q+1-\Tr(A)-2m)+m \sum_{i=1}^{k_1} \pi_i+k_2(\ell-1),
\end{equation}
where $\pi_i > \ell$ and where $\sum_{i=1}^{k_1}\pi_i$ is supposed to be zero if $k_1=0$.
\end{lemma}
 \begin{proof}
In order to prove the statement, let us recall that by Theorem 4 of \cite{Haloui} the number of rational points on an irreducible curve $D_i$ on $A$ of arithmetic genus $\pi_i$ satisfies $\#D_i(\mathbb{F}_q)\leq q+1-\Tr(A)+m|\pi_i -2|$. Since $A$ is simple and hence $\pi_i\geq 2$, we get $\#D_i(\mathbb{F}_q)\leq q+1-\Tr(A)-2m+m\pi_i$. 
Without loss of generality we consider $\{D_1,\ldots,D_{k_1}\}$  to be the set of the $D_i$ which have arithmetic genus $\pi_i > \ell$ and 
$\{D_{k_1+1},\ldots,D_{k}\}$ to be the set of the $k_2$ curves which have arithmetic genus $\pi_i \leq \ell$.
Thus, we get 
$$\sum_{i=1}^{k_1} \#D_i(\mathbb{F}_q)\leq k_1(q+1-\Tr(A)-2m)+m\sum_{i=1}^{k_1}\pi_i $$
where $\pi_i > \ell$.
Under the hypothesis that any absolutely irreducible curve on $A$ has arithmetic genus $>\ell$, we have that the $k_2$ curves that have arithmetic genus $\pi_i \leq \ell$ are necessarily non absolutely irreducible. It is well-known (see for example the proof of Theorem 4 of \cite{Haloui}) that if $D_i$ is a non absolutely irreducible curve of arithmetic genus $\pi_i$ lying on an abelian surface, its number of rational points satisfies $\#D_i(\mathbb{F}_q)\leq \pi_i-1$. Hence summing on $k_2$ we get
$$\sum_{i=k_1+1}^{k} \#D_i(\mathbb{F}_q)\leq\sum_{i=k_1+1}^{k} (\pi_i-1)\leq k_2(\ell-1).$$
The proof is now complete using inequality (\ref{Nf_sum}).
 \end{proof}

In order to use inequality (\ref{Nf_simple}) to deduce a lower bound on the minimum distance of the code $\mathcal{C}(A,rH)$, it is sufficient to bound the numbers $k_1$ and $k_2$ and the sum $\sum_{i=1}^{k_1} \pi_i$.

\begin{lemma}\label{SafiaGeneralise}
With the same notations and under the same hypotheses as Lemma \ref{Borne_N(f)} we have:
\begin{enumerate}[(1)]
\item\label{first} $k_1\sqrt{\ell}+k_2\leq r\sqrt{\frac{H^2}{2}},$
\item\label{third} 
$\sum^{k_1}_{i=1}\pi_i \leq\ \alpha^2+2\sqrt{\ell}\alpha+(\ell+1)k_1$, where
$\alpha := r\sqrt{\frac{H^2}{2}}-k_1\sqrt{\ell}-k_2$.
\end{enumerate}
\end{lemma}

\begin{proof}
Let us prove the first assertion. 
Since $H$ is ample, by Nakai-Moishezon criterion we have that 
$D_i.H>0$ for every $i=1,\dots,k$ and $H^2>0$. Thus,
we can
take the square root of inequality (\ref{sum_pi}) in the proof of Lemma \ref{HodgeTGV} and get $\sqrt{\pi_i-1} \leq D_i.H/ \sqrt{2H^2}$.
Now taking into account that $1\leq \pi_i-1$ since $A$ is assumed to be simple, summing for $i\in\{1,\dots ,k \}$, using that $n_i>0$ and that $\sum^k_{i=1} n_iD_i.H=rH^2$, we obtain
\begin{align*}
\sum^{k_1}_{i=1} \sqrt{\pi_i-1}
&=
\sum^{k}_{i=1} \sqrt{\pi_i-1}-\sum^{k}_{i=k_1+1} \sqrt{\pi_i-1} \\
&\leq  
\sum^{k}_{i=1} \sqrt{\pi_i-1}-k_2\\
&\leq  
\frac{1}{\sqrt{2H^2}} \sum^k_{i=1} n_i D_i.H-k_2\\
&=r\sqrt{\frac{H^2}{2}}-k_2.
\end{align*}
Considering the $k_1$ curves that have arithmetic genus $\pi_i > \ell$, we have $\sqrt{\ell} \leq \sqrt{\pi_i-1}$ and so $k_1\sqrt{\ell} \leq\sum^{k_1}_{i=1} \sqrt{\pi_i-1}$. Thus we get 
$$k_1\sqrt{\ell}+k_2\leq r\sqrt{\frac{H^2}{2}}.$$
Let us now prove the last statement. For $i=1,\dots,k_1$, set $s_i=\sqrt{\pi_i-1}-\sqrt{\ell}.$
Under the hypothesis that $\pi_i\geq \ell+1$, the $s_i$ are non-negative real numbers. Thus
$$\sum^{k_1}_{i=1} s_i^2\leq \left(\sum^{k_1}_{i=1} s_i\right)^2.$$
Moreover, we have seen above that
$$\sum^{k_1}_{i=1} s_i= \sum^{k_1}_{i=1} \sqrt{\pi_i-1}-k_1\sqrt{\ell}
 \leq r\sqrt{\frac{H^2}{2}}-k_1\sqrt{\ell}-k_2=\alpha.$$
Therefore, since $\pi_i=(s_i+\sqrt{\ell})^2+1=s_i^2+2s_i\sqrt{\ell}+\ell+1$ for $i\in\{1,\dots,k_1\}$, we have
\begin{equation} \label{FinLemme32}
\begin{split}
\sum^{k_1}_{i=1}\pi_i &=\sum_{i=1}^{k_1} s_i^2 +2\sqrt{\ell}\sum_{i=1}^{k_1} s_i +(\ell +1)k_1  \\
&\leq \left( \sum_{i=1}^{k_1} s_i\right)^2 +2\sqrt{\ell}\sum_{i=1}^{k_1} s_i +(\ell +1)k_1\\
&\leq \alpha^2+2\sqrt{\ell}\alpha+(\ell +1)k_1,
\end{split}
\end{equation}
which completes the proof of the lemma.
\end{proof}
We can now prove the following theorem.
\begin{theorem}\label{minimum_distance_simple_surfaces}
Let $A$ be a simple abelian surface defined over $\mathbb{F}_q$ of trace $\Tr(A)$.
Let $m=\lfloor 2\sqrt{q}\rfloor$, $H$ be an ample divisor on $A$
rational over $\mathbb{F}_q$ and $r$ be a positive integer large enough  so that $rH$ is very ample.
Moreover, let $\ell$ be a positive integer such that for every absolutely irreducible curves of arithmetic genus $\pi$ lying on $A$ we have $\pi>\ell$.
Then the minimum distance $d(A,rH)$ of the code $\mathcal{C}(A,rH)$ satisfies
$$
d(A,rH)\geq  \#A(\mathbb{F}_q)-\max\left( \left\lfloor r\sqrt{\frac{H^2}{2}}\right\rfloor(\ell-1), \varphi(1), \varphi\left(\left\lfloor r\sqrt{\frac{H^2}{2\ell}}\right\rfloor \right)\right),$$
where
\footnotesize{
$$\varphi(x):= m\left(r \sqrt{\frac{H^2}{2}} - x \sqrt{\ell}\right)^2+2m\sqrt{\ell}\left(r \sqrt{\frac{H^2}{2}} - x \sqrt{\ell}\right)  +x \Big{(} q+1-\Tr(A)+(\ell-1)(m-\sqrt{\ell})\Big{)} + r \sqrt{\frac{H^2}{2}}(\ell-1).$$}
\end{theorem}

\begin{proof}
Recall that $d(A,rH)= \#A(\mathbb{F}_q)-\max \{N(f), f\in L(rH)\setminus \{0\} \} .$
The point of departure is the inequality (\ref{Nf_simple}).
When $k_1=0$, it simply implies that $N(f)$ is 
 less than or equal to $k_2 (\ell-1)$.
If $k_1>0$  we use point (\ref{third}) of Lemma \ref{SafiaGeneralise} to get
\begin{equation}\label{inegalite_amelioree}
\footnotesize{
N(f) \leq 
k_1 \Big{(} q+1-\Tr(A)+(\ell-1)\left(m-\sqrt{\ell}\right)
\Big{)} + m\alpha^2+2m\sqrt{\ell}\alpha   + r \sqrt{\frac{H^2}{2}}(\ell-1)}
-\alpha (\ell - 1) 
\end{equation}
where we have set $\alpha:=r \sqrt{\frac{H^2}{2}} - k_1 \sqrt{\ell} - k_2$.
Now point (\ref{first}) of Lemma \ref{SafiaGeneralise} ensures that $\alpha \geq 0$
which enables to conclude that $N(f) \leq \phi(k_1,k_2)$ where $\phi(k_1,k_2)$
is defined by
\begin{equation*}
    \phi(k_1,k_2) = \begin{cases}
               k_2 (\ell-1)           & \text{if } k_1 = 0\\ \ \\
               
                k_1 \Big{(} q+1-\Tr(A)+(\ell-1)\left(m-\sqrt{\ell}\right)\Big{)}                   &  \\
               + m\alpha^2+2m\sqrt{\ell}\alpha   + r \sqrt{\frac{H^2}{2}}(\ell-1) & \text{if } k_1>0.

           \end{cases}
\end{equation*}
It thus remains to
maximise the function $\phi(k_1, k_2)$ on the integer points 
inside
 the polygon ${\mathcal K}$ defined by
$${\mathcal K}=\left\{(k_1,k_2)\mid 0\leq k_1, 0\leq k_2, 1\leq k_1+k_2, \sqrt{\ell} k_1+k_2 \leq r\sqrt{\frac{H^2}{2}}\right\}$$
and represented below.

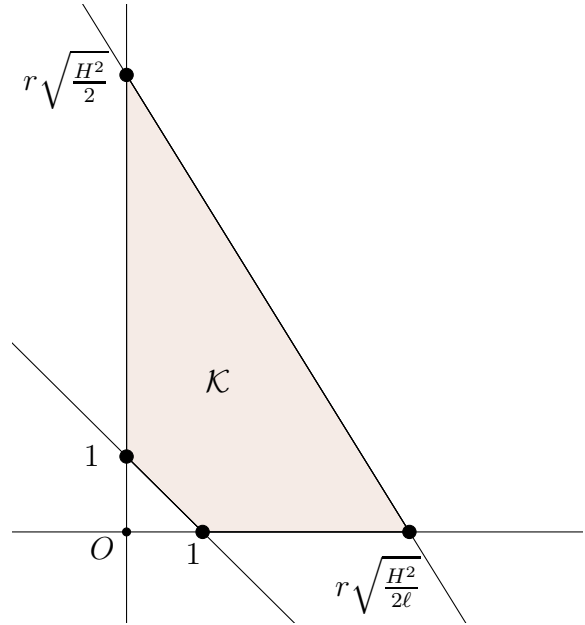
\begin{center}
\definecolor{zzttqq}{rgb}{0.6,0.2,0.}
\begin{tikzpicture}[line cap=round,line join=round,=triangle 45,x=1.0cm,y=1.0cm]
\clip(-1.5,-1.2) rectangle (6,7);
\fill[color=zzttqq,fill=zzttqq,fill opacity=0.10000000149011612] (0.,6.06) -- (0.,1.) -- (1.,0.) -- (3.72,0.) -- cycle;
\draw [domain=-2.4:13.38] plot(\x,{(--1.-1.*\x)/1.});
\draw [domain=-2.4:13.38] plot(\x,{(--22.5432-6.06*\x)/3.72});
\draw (1.,0.)-- (3.72,0.);
\draw (3.72,0.)-- (0.,6.06);
\draw (0.,6.06)-- (0.,1.);
\draw (0.,1.)-- (1.,0.);
\draw (0.,6.06)-- (0.,1.);
\draw (0.,1.)-- (1.,0.);
\draw (1.,0.)-- (3.72,0.);
\draw (3.72,0.)-- (0.,6.06);
\draw (0.,-2.22) -- (0.,7.12);
\draw [domain=-2.4:13.38] plot(\x,{(-0.-0.*\x)/2.72});
\draw [fill=black] (0.,1.) circle (2.5pt);
\draw[color=black] (-0.46,1.01) node {$1$};
\draw [fill=black] (1.,0.) circle (2.5pt);
\draw[color=black] (0.88,-0.29) node {$1$};
\draw [fill=black] (3.72,0.) circle (2.5pt);
\draw[color=black] (1.2,2) node {${\mathcal K}$};
\draw[color=black] (3.3,-0.7) node {$r\sqrt{\frac{H^2}{2\ell}}$};
\draw [fill=black] (0.,6.06) circle (2.5pt);
\draw[color=black] (-0.8,6) node {$r\sqrt{\frac{H^2}{2}}$};
\draw [fill=black] (0.,0.) circle (1.5pt);
\draw[color=black] (-0.32,-0.21) node {$O$};
\end{tikzpicture}
\end{center}
\medskip
First, in the case where $k_1=0$ we notice that
 $k_2\leq \left\lfloor r\sqrt{\frac{H^2}{2}}\right\rfloor$,
which  implies $\phi(0,k_2) \leq  \left\lfloor r\sqrt{\frac{H^2}{2}}\right\rfloor(\ell-1)$.
Second, for a fixed positive value of $k_1$ less than or equal to $r\sqrt{\frac{H^2}{2}}$
we can consider
$\phi$
as a degree-2 polynomial in $\alpha \geq 0$, namely
$\phi(k_1, k_2)= m\alpha(\alpha+2\sqrt{\ell})+\hbox{constant}$.
This way, it is clear that the maximum of $\phi$   
is reached for the maximal value of $\alpha$,
that is for the minimal value of $k_2$ such that $(k_1, k_2) \in {\mathcal K}$.
Hence, for this second case we are reduced to maximise $\phi$ on 
the segment $\left[(1,0), \left(r\sqrt{\frac{H^2}{2\ell}},0\right)\right]$.
As easily checked, $\phi$ is a convex function on this segment,
so the maximum 
 is reached at an extremal integer point, $(1,0)$ or $\left(\left\lfloor r\sqrt{\frac{H^2}{2\ell}}\right\rfloor,0\right)$.
Finally, note that we have $\phi(x,0)=\varphi(x)$, and the theorem is proved.
\end{proof}

\begin{remark} Taking into account the term
$-\alpha(\ell-1)$ in the inequality (\ref{inegalite_amelioree})
one sometimes obtains a slightly better bound than the one of Theorem \ref{minimum_distance_simple_surfaces}
  but whose expression is even more complicated.
\end{remark}

\begin{remark}\label{comparing_with_Haloui}
The bound in Theorem \ref{minimum_distance_simple_surfaces} applies with $\ell=1$ on simple abelian surfaces since they do not contain absolutely irreducible curves of arithmetic genus $0$ nor $1$, as remarked at the beginning of this section. Note that for $\ell=1$ we have $\left\lfloor r\sqrt{\frac{H^2}{2}}\right\rfloor (\ell-1)=0$ and thus in this context we are reduced to consider the maximum between $\varphi(1)$ and $\varphi\left(\left\lfloor r\sqrt{\frac{H^2}{2}}\right\rfloor \right)$ in Theorem \ref{minimum_distance_simple_surfaces}. In order to easily compare these two values, let us consider a weaker version of our theorem by removing the integer part. Indeed, $\varphi\left(\left\lfloor r\sqrt{\frac{H^2}{2}}\right\rfloor\right)\leq\varphi\left(r\sqrt{\frac{H^2}{2}}\right)$. Consequently we have $d(A,rH)\geq  \#A(\mathbb{F}_q)-\max\left(\varphi(1),\varphi\left(r\sqrt{\frac{H^2}{2}}\right)\right)$ and after some calculations we get
\begin{equation*}
d(A,rH)\geq 
\begin{cases} \#A(\mathbb{F}_q)-r\sqrt{\frac{H^2}{2}}\left(q+1-\Tr(A)\right) \text{ if } r\leq \frac{\sqrt{2}(q+1-\Tr(A)-m)}{m\sqrt{H^2}},\\
 \#A(\mathbb{F}_q)-(q+1-\Tr(A)-m)-mr^2\frac{H^2}{2} \text{ otherwise}.
\end{cases}
\end{equation*}
In particular if $A=\Jac(C)$ is the Jacobian of a curve $C$ of genus $2$ which is simple, then by setting $H=C$ with $H^2=C^2=2\pi_C-2=2$ by the adjunction formula, we obtain 
\begin{equation*}
d(\Jac(C),rC)\geq 
\begin{cases} \#\Jac(C)(\mathbb{F}_q)-r\# C(\mathbb{F}_q) \text{ if } r\leq \frac{q+1-\Tr(A)-2m}{m},\\
 \#\Jac(C)(\mathbb{F}_q)-\# C(\mathbb{F}_q)-m(r^2-1)
 \text{ otherwise}.
\end{cases}
\end{equation*}

This bound coincides with the bound in the main theorem of \cite{Haloui}.
\end{remark}

\begin{remark}\label{mieux_ell}
We point out, using an elementary asymptotic analysis for large $q$ and $r$, that our estimation of the minimum distance is better for larger $\ell$.
We assume 
that $\ell$ is small (for example $\ell$ is a fixed value)
and that  $r=q^{\rho}$ for some $\rho>0$. 
For simplicity, we also assume that $H^2=2$ (see Section~\ref{explicit}) and remove the integer part.
Taking into account that $\vert \Tr(A)\vert \leq 4\sqrt{q}$ yields to
\begin{equation*}
\left\{
\begin{matrix}
r(\ell -1)\sqrt{\frac{H^2}{2}} &\underset{q\to\infty}{\sim} &(\ell-1)  q^{\rho},\\
\varphi(1) &\underset{q\to\infty}{\sim} & c q^{\max\{1, 2\rho+\frac{1}{2}\}},\\
\varphi\left(r\sqrt{\frac{H^2}{2\ell}}\right) &\underset{q\to\infty}{\sim} &\frac{1}{\sqrt{\ell}} q^{1+\rho},
\end{matrix}
\right.
\end{equation*}
where $c=1,3$ or $2$ depending on whether $\rho<1/4, \rho=1/4$ or $\rho >1/4$.
Consequently, in this setting, the lower bound $d^*$ obtained in Theorem \ref{minimum_distance_simple_surfaces} satisfies
\begin{equation*}
\begin{cases} \#A(\mathbb{F}_q)-d^* \underset{q\to\infty}{\sim} 2q^{2\rho+\frac{1}{2}} \text{ if } \rho\geq\frac{1}{2},\\
 \#A(\mathbb{F}_q)-d^* \underset{q\to\infty}{\sim} \frac{1}{\sqrt{\ell}} q^{1+\rho} \text{ if } 0<\rho<\frac{1}{2}.
 \end{cases}
\end{equation*}
So for $q$ sufficiently large and $r=q^{\rho}$ with $0<\rho<\frac{1}{2}$, the bound in Theorem~\ref{minimum_distance_simple_surfaces} obtained for $\ell=2$ 
for instance is better than the one obtained for $\ell=1$,
 that is for any simple abelian variety. We thus focus in the next section on the existence of simple abelian surfaces which do not contain absolutely irreducible curves of arithmetic genus $2$.
 \end{remark}


\section{Abelian surfaces without curves of genus $1$ nor $2$}\label{Weil_restriction_of_elliptic_curves}
In light of Remark \ref{mieux_ell}, 
 considering abelian surfaces 
 without
  absolutely irreducible curves of small arithmetic genus will lead to a sharper lower bound on the minimum distance of the evaluation code $\mathcal{C}(A,rH)$.
Hence in this section we look for abelian surfaces which satisfy the property not to contain 
 absolutely irreducible curves defined over $\mathbb{F}_{q}$ of arithmetic genus $0$, $1$ nor $2$.
 
By the theorem of classification of Weil  (see for instance \cite[Th.1.3]{Howe_Nart_Ritzenthaler}),
a principally polarized abelian surface defined over $\mathbb{F}_q$ is isomorphic to either the polarized Jacobian of a curve of genus 2 over $\mathbb{F}_q$,
either the product of two polarized elliptic curves over $\mathbb{F}_q$ or either the Weil restriction from $\mathbb{F}_{q^2}$ to $\mathbb{F}_q$ of 
a polarized elliptic curve defined over $\mathbb{F}_{q^2}$.
It is straightforward to see that  the Jacobian of a curve of genus $2$ contains the curve itself and that the product of two elliptic curves contains copies of each of them. 

It therefore remains two cases to consider.
First, there is the case of 
abelian surfaces which do not admit a principal polarization.
We prove in Proposition \ref{nppas} that they always satisfy the desired property.
Secondly, we give
in Proposition \ref{without_curves}
necessary and sufficient conditions 
for Weil restrictions of elliptic curves
to satisfy the same property.

Throughout this section we will make use of the two following well-known results. An abelian surface contains a {\sl smooth} absolutely irreducible curve of genus $1$ if and only if it
is isogenous to the product of two elliptic curves.
Moreover, a simple abelian surface contains a {\sl smooth} absolutely irreducible curve of genus $2$ if and only if it is isogenous to the Jacobian of a curve of genus $2$ (see \cite[Proposition 2]{Galbraith_Smart}). 
The following lemma gives necessarily and sufficient conditions to avoid the presence of 
 {\sl non necessarily smooth} absolutely irreducible curves of low arithmetic genus.
\begin{lemma} \label{petit_pi_petit_g} 
Let $A$ be an abelian surface. Then the three following statements are equivalent:
\begin{enumerate}[(1)]
\item\label{uno} $A$ is simple and not isogenous to a Jacobian surface;
\item\label{due} $A$ does not contain absolutely irreducible curves of arithmetic genus $0$, $1$ nor $2$;
\item\label{tre} $A$ does not contain absolutely irreducible smooth curves of genus $0$, $1$ nor $2$.
\end{enumerate}
\end{lemma}

\begin{proof}
Let us prove that (\ref{uno}) $\Rightarrow$ (\ref{due}).
Let A be a simple abelian surface which is not isogenous to the Jacobian of a curve of genus $2$. Let $C$ be an absolutely irreducible curve lying on $A$  and let $\nu : \tilde{C} \mapsto C$ be its normalisation map. The case of genus $0$ and $1$ is treated in \cite[\S2]{Haloui}. 
For the genus $2$ case, assume by contradiction that $\pi(C)=2$. We get $g(\tilde C)=\pi(C)=2$ so $\tilde{C}=C$ is smooth
 and thus by Proposition 2 of \cite{Galbraith_Smart} $A$ is isogenous to the Jacobian of $C$, in contradiction with the hypotheses.
 
The implication (\ref{due}) $\Rightarrow$ (\ref{tre}) is trivial since for smooth curves the geometric and arithmetic genus coincide.
 
Finally let us prove that (\ref{tre}) $\Rightarrow$ (\ref{uno}). Assume by contradiction that $A$ is not simple, hence $A$ is isogenous to the product of two elliptic curves and thus it contains at least a smooth absolutely irreducible curve of genus $1$, in contradiction with (\ref{tre}). Now assume that $A$ is simple and isogenous to a Jacobian surface. Then by Proposition 2 of \cite{Galbraith_Smart}, $A$ contains a smooth absolutely irreducible curve of genus $2$, again in contradiction with (\ref{tre}). This concludes the proof.
 
\end{proof}

\subsection{Non-principally polarized abelian surfaces}\label{NonPrincPo}
An isogeny class of abelian varieties over ${\mathbb F}_q$ is said to be not principally polarizable if it does not contain a principally polarizable abelian variety over ${\mathbb F}_q$.
The following proposition states that abelian surfaces 
which do not admit a principal polarization
have naturally the property we are searching for.
\begin{proposition}\label{nppas}
Let $A$ be an abelian surface in a not principally polarizable isogeny class. Then $A$ does not contain absolutely irreducible curves of arithmetic genus $0$, $1$ nor $2$.
\end{proposition}
\begin{proof}
It is well-known that an abelian variety contains no curves of genus $0$. 
Since $A$ is not isogenous to a principally polarizable abelian surface, it follows that it is not isogenous to a product of two  elliptic curves nor to a Jacobian surface.
By Lemma \ref{petit_pi_petit_g} we conclude the proof.
\end{proof}

To be concrete, let us recall here a characterisation 
of non-principally polarized isogeny class of abelian surfaces (\cite[Th.1]{Howe_Maisner_Nart_Ritzenthaler}) 
for which Theorem \ref{minimum_distance_simple_surfaces} applies with $\ell=2$.
An isogeny class of abelian surfaces defined over $\mathbb{F}_q$ with Weil polynomial $f(t)=t^4+at^3+bt^2+qat+q^2$ is not principally polarizable if and only if the following three conditions are satisfied:
\begin{enumerate}[(1)]
\item $a^2-b=q$;
\item $b<0$;
\item all prime divisors of $b$ are congruent to $1\mod 3$. 
\end{enumerate}

\subsection{Weil restrictions of elliptic curves}\label{WeilRest}
Let $k=\mathbb{F}_q$ and $K$ denotes an extension of finite degree $[K:k]$ of $k$. Let $E$ be an elliptic curve defined over $K$. The $K/k$-Weil restriction of scalars of $E$ is an abelian variety $W_{K/k}(E)$ of dimension $[K:k]$ defined over $k$ (see \cite[\S16]{Milne_AG} for a presentation in terms of universal property and see \cite[\S3]{Galbraith_Smart} for a constructive approach). We consider here the $\mathbb{F}_{q^2}/\mathbb{F}_q$-Weil restriction of an elliptic curve $E$ defined over $\mathbb{F}_{q^2}$ which is an abelian surface $A$ defined over $\mathbb{F}_q$.

Let $f_{E/\mathbb{F}_{q^2}}(t)$ be the Weil polynomial of the elliptic curve $E$ defined over $\mathbb{F}_{q^2}$. Then the Weil polynomial of $A$ over $\mathbb{F}_q$  is given (see  \cite[Prop 3.1]{Galbraith}) by
\begin{equation}\label{lien_f}
f_{A/\mathbb{F}_{q}}(t)=f_{E/\mathbb{F}_{q^2}}(t^2).
\end{equation}
Since $f_{E/\mathbb{F}_{q^2}}(t)=t^2-\Tr(E/\mathbb{F}_{q^2})t+q^2$ we have $f_A(t)=t^4-\Tr(E/\mathbb{F}_{q^2})t^2+q^2$, thus it follows from (\ref{weil_polynomial}) that the trace of $A$ over $\mathbb{F}_q$ is equal to $0$.
Moreover, since the number of $\mathbb{F}_q$-rational points on an abelian variety $A$ defined over $\mathbb{F}_q$ equals $f_{A/\mathbb{F}_{q}}(1)$, we get that the number of rational points on $A=W_{\mathbb{F}_{q^2}/\mathbb{F}_q}(E)$ over $\mathbb{F}_q$ is the same as the number of rational points on $E$ over $\mathbb{F}_{q^2}$, i.e.\ we have $\# A (\mathbb{F}_q) =f_{A/\mathbb{F}_{q}}(1)=f_{E/\mathbb{F}_{q^2}}(1)= \# E(\mathbb{F}_{q^2}).$

\medskip

\begin{proposition}\label{without_curves}
Let $q$ be a power of a prime $p$. Let $E$ be an elliptic curve defined over $\mathbb{F}_{q^2}$ of Weil polynomial $f_{E/\mathbb{F}_{q^2}}(t)=t^2-\Tr(E/\mathbb{F}_{q^2})t+q^2$. Let $A$ be the $\mathbb{F}_{q^2}/\mathbb{F}_{q}$-Weil restriction of the elliptic curve $E$. Then $A$ does not contain absolutely irreducible curves defined over $\mathbb{F}_{q}$ of arithmetic genus $0$, $1$ nor $2$ if and only if one of the following conditions holds
\begin{enumerate}[(1)]
\item\label{one} $\Tr(E/\mathbb{F}_{q^2})=2q-1$;
\item\label{two} $p>2$ and $\Tr(E/\mathbb{F}_{q^2})=2q-2$;
\item $p \equiv 11 \mod 12$ or $p=3$, $q$ is a square and $\Tr(E/\mathbb{F}_{q^2})=q$;
\item $p=2$, $q$ is nonsquare and $\Tr(E/\mathbb{F}_{q^2})=q$;
\item\label{five} $q=2$ or $q=3$ and $\Tr(E/\mathbb{F}_{q^2})=2q$.
\end{enumerate}
\end{proposition}

\begin{proof}
 Let $E$ be an elliptic curve defined over $\mathbb{F}_{q^2}$ and let $A$ be the $\mathbb{F}_{q^2}/\mathbb{F}_{q}$-Weil restriction of $E$.
Let $f_A(t)=t^4+at^3+bt^2+qat+q^2$ be the Weil polynomial of $A$. Recall that we have $f_A(t)=t^4-\Tr(E/\mathbb{F}_{q^2})t^2+q^2$ by~(\ref{lien_f}) and thus $(a, b)=(0, -\Tr(E/\mathbb{F}_{q^2}))$. Theorem 1.2-(2) with Table 1.2 in \cite{Howe_Nart_Ritzenthaler} gives necessary and sufficient conditions on the couple $(a,b)$ for a simple abelian surface with the corresponding Weil polynomial not to be  isogenous to the Jacobian of a smooth curve of genus $2$.

Let us suppose that the trace of the elliptic curve $E$ over $\mathbb{F}_{q^2}$ does not  fit one of the conditions $(\ref{one}) - (\ref{five})$. 
Let us remark that by Theorem 1.4  in  \cite{Howe_Nart_Ritzenthaler} the first case of Table 1.2 in \cite[Theorem 1.2-(2)]{Howe_Nart_Ritzenthaler} corresponds to all simple abelian surfaces which do not admit a principal polarization.
Moreover the cases $(\ref{one}) - (\ref{five})$ cover the remaining cases of Table 1.2.
Then  $f_A(t)$ does not represent an isogeny class of simple principally polarizable abelian surfaces not containing a Jacobian surface. 
Hence $A$ is either not principally polarizable, or not simple or isogenous to the Jacobian of a curve of genus $2$. In the first case $A$ would not be a Weil restriction of an elliptic curve since these last one admit a principal polarization. In the second case, $A$ would contain a curve of genus $1$  and finally in the third case it would contain a curve of genus $2$. 
Thus we proved that if $A$ does not contain absolutely irreducible curves defined over $\mathbb{F}_{q}$ of arithmetic genus $0$, $1$ nor $2$ then one of conditions $(\ref{one}) - (\ref{five})$ holds.

Conversely, using again Table 1.2 in \cite[Theorem 1.2-(2)]{Howe_Nart_Ritzenthaler} we get that in each case from $(\ref{one})$ to $(\ref{five})$ of our proposition, the couple $(0, -\Tr(E/\mathbb{F}_{q^2}))$ corresponds to simple abelian surfaces not isogenous to the Jacobian of a curve of genus $2$. Therefore in these cases $A$ does not contain absolutely irreducible {\sl smooth} curves of geometric genus $0$, $1$ nor $2$, and thus by Lemma \ref{petit_pi_petit_g}, $A$ does not contain absolutely irreducible curves of arithmetic genus $0$, $1$ nor $2$.
\end{proof}

\begin{remark}
Let us mention two cases in which Weil restrictions of elliptic curves do contain curves of genus $1$ or $2$.
First, if the elliptic curve $E$ is defined over $\mathbb{F}_q$, it is clearly  a subvariety of $A$.
Note that in Proposition \ref{without_curves} we do not need to suppose that the elliptic curve $E$ defined over $\mathbb{F}_{q^2}$ is not defined over $\mathbb{F}_q$ because 
none of the the elliptic curves with trace over $\mathbb{F}_{q^2}$ as in cases (\ref{one})-(\ref{five}) is defined over $\mathbb{F}_q$.
Secondly, it is well-known that there are Weil restrictions of elliptic curves 
that are isogenous to Jacobian surfaces (see for example \cite{Scholten}) 
which thus contain smooth curves of genus $2$. 
\end{remark}

\begin{remark}
Let $q^2=p^{2n}$ with $p$ prime. By Deuring theorem (see for instance \cite[Th. 4.1]{Waterhouse}) for every integer $\beta$ satisfying $|\beta|\leq 2q$ such that $\gcd(\beta, p)=1$, or $\beta=\pm 2q$, or $\beta=\pm q$ and $p\not\equiv 1 \mod 3$, there exists an elliptic curve of trace $\beta$ over $\mathbb{F}_{q^2}$. Using Deuring theorem it is easy to check the existence of an elliptic curve with the given trace for each of the five cases in the previous theorem. 
\end{remark}

\begin{remark}
Let us remark that the first bound in Theorem \ref{minimum_distance_simple_surfaces} becomes relevant for $q\geq B$ with $B\approx 4(\sqrt{H^2}+1)^2$ and it is non-relevant for small $q$. Therefore case (\ref{five}) of Proposition \ref{without_curves} does not give rise to practical cases.
\end{remark}

Let us briefly outline the results obtained in the last sections.  The surfaces arising in Propositions  \ref{nppas} and  \ref{without_curves} give rise to codes for which the lower bound on the minimum distance of Theorem \ref{minimum_distance_simple_surfaces} applies with $\ell=2$.
This is exactly the purpose of the proposition stated in the introduction.

We have exploited the fact that,
for $q$ sufficiently large and $r=q^{\rho}$ with $0<\rho<\frac{1}{2}$,
the bound obtained for $\ell=2$ improves the one obtained for $\ell=1$.
Note also that under the same hypotheses the bound for $\ell=3$ improves the one for $\ell=2$. Hence it would be interesting in the future to investigate on the existence of abelian surfaces  without absolutely irreducible curves of genus $\leq 3$ lying on them.

\section{To make explicit the lower bounds for the minimum distance}\label{explicit}
We now show how the terms $\# A(\mathbb{F}_q)$, $\Tr(A)$ and $H^2$ appearing in the lower bounds for the minimum distance $d(A,rH)$ given in Theorems \ref{ourbound} and \ref{minimum_distance_simple_surfaces} can be computed in many cases. As already said in the introduction of Section \ref{Weil_restriction_of_elliptic_curves}, three cases have to be distinguished in the case of principally polarized abelian surfaces, according to Weil classification. 

Let $A$ be a principally polarized abelian surface defined over ${\mathbb F}_q$ with Weil polynomial
$f_A(t)=(t-\omega_1)(t-\overline{\omega}_1)(t-\omega_2)(t-\overline{\omega}_2)$ where the $\omega_i$'s are complex numbers of modulus $\sqrt{q}$.
Then we get:
$$\# A(\mathbb{F}_q)=f_A(1)=(1-\omega_1)(1-\overline{\omega}_1)(1-\omega_2)(1-\overline{\omega}_2).$$
From formula (\ref{weil_polynomial}), we obtain:
$$\Tr(A)= \omega_1+\overline{\omega}_1+\omega_2+\overline{\omega}_2.$$

Moreover, for any divisor $H$ on $A$, the adjunction formula gives  
$$H^2=2\pi_H-2.$$

As recalled in Section \ref{Codes_from_Abelian_Surfaces}, if the divisor $H$ is ample then $rH$ is very ample as soon as $r\geq 3$.

\medskip

\begin{enumerate}[(1)]

\item In case $A$ is the Jacobian $\Jac(C)$ of a genus-2 curve $C$ defined over ${\mathbb F}_q$, 
the numerator $P_C(t)$ of the zeta function of $C$ is equal to the reciprocal polynomial of the Weil polynomial
 $f_{\Jac(C)}(t)$:
$$P_C(t)=t^{4} f_{\Jac(C)}\left(\frac{1}{t}\right)=(1-\omega_1t)(1-\overline{\omega}_1t)(1-\omega_2t)(1-\overline{\omega}_2t).$$
Hence we obtain
\begin{equation*}
\left\{
\begin{matrix}
\# C(\mathbb{F}_q) &=& q+1- (\omega_1+\overline{\omega}_1+\omega_2+\overline{\omega}_2) \\
\# C(\mathbb{F}_{q^2}) &=& q^2+1- (\omega_1^2+\overline{\omega}_1^2+\omega_2^2+\overline{\omega}_2^2)
\end{matrix}
\right.
\end{equation*}

and thus
$$\# \Jac(C)({\mathbb F}_q)=\frac{1}{2}(\# C({\mathbb F}_{q^2})+\# C({\mathbb F}_q)^2)-q.$$

\medskip

Now, choosing $H=C$ for instance, we get an ample divisor with $H^2=C^2=2\pi_C-2=2$.

\bigskip

\item  In case 
$A$ is the product $E_1\times E_2$ of two elliptic curves $E_1$ and $E_2$
each partial trace $\Tr(E_i)=\omega_i+ \overline{\omega}_i$ is determined by $\# E_i(\mathbb{F}_q) = q+1-\Tr(E_i)$.
So we have 
$\# A(\mathbb{F}_q)= \#E_1(\FF_q)\times \#E_2(\FF_q)$ and $\Tr(A)= \Tr(E_1)+\Tr(E_2)$.

\medskip

Any choice of rational points $P_i\in E_i$ leads to an ample divisor
$H=E_1\times\{P_2\}+\{P_1\}\times E_2$
 such that
$H^2=(E_1\times\{P_2\})^2+(\{P_1\}\times E_2)^2+2(E_1\times\{P_2\}).(\{P_1\}\times E_2)=0+0+2\times 1=2$.

\bigskip

\item  In the last case where $A=W_{\mathbb{F}_{q^2}/\mathbb{F}_q}(E)$ is the $\mathbb{F}_{q^2}/\mathbb{F}_q$-Weil restriction of an elliptic curve $E$ defined over $\mathbb{F}_{q^2}$, then we have already seen in Subsection \ref{WeilRest} that $\# A(\mathbb{F}_q)=\# E(\mathbb{F}_{q^2})$ and that $\Tr(A)=0$.

\medskip

As an ample divisor on $A$, one can choose for instance $H=E+E^q$ where  $E^q$ is the image of $E$ by the generator $\sigma:x\longmapsto x^q$ of the Galois group $\Gal(\mathbb{F}_{q^2}/\mathbb{F}_q)$. We 
thus  have $H^2=E^2+(E^q)^2+2E.E^q=0+0+2\times 1=2$.
\end{enumerate} \ \\ \ \\
\noindent
{\bf Acknowledgments:} The authors would like to thank the anonymous referee for relevant observations and in particular for advices which have led to a sharper form of Theorem~\ref{minimum_distance_simple_surfaces}.
\nocite{*} \\ \ \\ \ \\
\noindent
\bigskip
\printbibliography

\Addresses

\end{document}